%% file: main.tex
\documentclass[12pt]{article}

\input{packages}

\input{macros}

\begin{document}
\pagenumbering{gobble} 

\input{title}

\newpage
\setcounter{tocdepth}{2}
\tableofcontents
\newpage
\pagenumbering{arabic}  

\input{introduction}
\input{preliminaries}
\input{Evaluations_that_Retain_Irreducibility}
\input{HSG_with_polynomial_evaluations}
\input{construction}
\input{reduction}

\input{ai}
\bibliographystyle{alpha}
\bibliography{bib}
\end{document}

%% file: packages.tex
\usepackage[utf8]{inputenc}
\usepackage[english]{babel}

\usepackage[a4paper, total={6.3in, 8.9in}]{geometry}

\usepackage{tocloft}
\usepackage{mathptmx} 

\usepackage{verbatim}
\usepackage{enumitem}
\usepackage{array,multirow}
\usepackage{enumitem}
\usepackage{amssymb,amsfonts,amsmath,amsthm}
\usepackage{aliascnt} 
\usepackage{url}
\usepackage{tikz}
\usepackage{bbm}
\usepackage{accents}
\usepackage{thmtools}

\usetikzlibrary{arrows.meta}
\usetikzlibrary{calc}

\usepackage{xcolor}
\usepackage{ctable}
\usepackage{tikz}
\usetikzlibrary{matrix,arrows}
\usetikzlibrary{decorations.markings}

\definecolor{citec}{HTML}{324FDA} 
\definecolor{linkc}{HTML}{A0111A}
\definecolor{urlc}{HTML}{b55c87}
\usepackage[colorlinks,
    citecolor=citec,
    linkcolor=linkc,
    urlcolor=urlc,
    bookmarks = true,
    breaklinks,
]{hyperref}
\usepackage[nameinlink, noabbrev, capitalize]{cleveref}

%% file: macros.tex
\theoremstyle{plain}
\newtheorem{theorem}{Theorem}[section]

\newtheorem{lemma}[theorem]{Lemma}
\newtheorem{corollary}[theorem]{Corollary}
\newtheorem{proposition}[theorem]{Proposition}
\newtheorem{claim}[theorem]{Claim}
\newtheorem{definition}[theorem]{Definition}

\newtheorem{remark}[theorem]{Remark}
\newtheorem{algorithm}[theorem]{Algorithm}

\numberwithin{equation}{section}

\newcommand{\eps}{\varepsilon}
\newcommand{\Tr}{\operatorname{Tr}}

\def\poly{{\mathrm{poly}}}

\DeclareMathOperator{\Syl}{Syl}

\renewcommand{\Pr}{\mathop{\bf Pr\/}}

\newcommand{\FF}{{\mathbb{F}}}
\newcommand{\KK}{{\mathbb{K}}}
\newcommand{\EE}{{\mathbb{E}}}

\newcommand{\coeff}[1]{\mathrm{coeff}\left(#1\right)}

\newcommand{\Res}[1]{\mathrm{Res}\left(#1\right)}

%% file: title.tex
\title{Optimal PRGs for Low-Degree Polynomials \\ over Polynomial-Size Fields\thanks{An earlier version of this paper appeared on ECCC \cite{CDG26}, predating the works of \cite{CHLLTV26,Cohen26}. The present revision remains independent of these works. The main change from the earlier version is the removal of the restriction on the characteristic.}}

\date{}

        \author{
  Gil Cohen\thanks{Tel Aviv University. \texttt{gil@tauex.tau.ac.il}. Supported by ERC starting grant 949499 and by the Israel Science Foundation grant 2989/24.  	
  }
  \and
  Noam Goldgraber\thanks{Ben Gurion University and Tel Aviv University. \texttt{goldgrab@post.bgu.ac.il}. Supported by NSF-BSF grant 2022644.}
}

\maketitle

\input{abstract}

%% file: abstract.tex
\begin{abstract}

We construct the first pseudorandom generator (PRG) with \emph{optimal seed length} for degree-$d$ polynomials over fields of size $q \approx d^4$, with no restriction on the characteristic. This improves on the construction of Dwivedi, Guo, and Volk~\cite{dwivedi2024optimalpseudorandomgeneratorslowdegree}, which requires fields of size exponential in~$d$ and sufficiently large characteristic, and on that of Derksen and Viola~\cite{DV22}, which requires a field size that grows with the number of variables~$n$.

A recent breakthrough by Chattopadhyay, Hatami, Lee, Lovett, Tal, and Viola~\cite{CHLLTV26} gives a PRG over the binary field with seed length $\widetilde{O}(d^2\log^2 n)$, with subsequent work by Cohen~\cite{Cohen26} achieving seed length $O(d^4\log n)$. Both constructions use polarized random walks, and the current analyses incur at least a quadratic dependence on $d$. Our second result is a reduction that provides an alternative route to nearly optimal seed length over the binary field through improved constructions over larger fields. Specifically, we prove that, for any $\delta \in (0,1)$, an optimal-seed PRG for degree-$d$ polynomials over fields of characteristic two and size $q \approx d^\delta$ would yield a PRG over the binary field with seed length $O(d^{1/(1-\delta)}\log n)$.

\end{abstract}

%% file: introduction.tex
\section{Introduction}

A \emph{pseudorandom generator} (PRG) fooling a class of functions $\mathcal{C} \subseteq \Sigma^n \rightarrow \Sigma$ is a map $G \colon \{0,1 \}^s \rightarrow \Sigma^n$ that stretches a uniform seed of $s \ll n$ bits into strings in $\Sigma^n$, such that the distribution of $G$ \emph{fools} any function $f \in \mathcal{C}$, in the sense that $f(G(\mathbf{U}_s))$ is close, in total-variation distance, to $f(\mathbf{U}_{\Sigma^n})$.\footnote{Here and throughout, for an integer $s$, $\mathbf{U}_s$ denotes the uniform distribution over $\{0,1\}^s$, and for a set $A$, $\mathbf{U}_A$ denotes the uniform distribution over $A$.} Constructing explicit PRGs with short seed for various function classes $\mathcal{C}$ is central to theoretical computer science, with various applications in complexity theory (prominently derandomization and circuit lower bounds), cryptography, and algorithm design.

A fundamental and well-studied class of functions is that of \emph{low degree polynomials}.

\begin{definition}
    We say that $G \colon \{0,1\}^s \rightarrow \FF_q^n$ is a PRG for $n$-variate polynomials of total degree at most $d$ over a finite field $\mathbb{F}_q$ with error $\varepsilon$ if for every such polynomial $f$, the distributions $f(G(\mathbf{U}_{s}))$ and $f(\mathbf{U}_{\mathbb{F}_q^n})$ are $\varepsilon$-close in total variation distance. That is,
    \[
    \frac{1}{2} 
    \sum_{a \in \mathbb{F}_q}
    \left|
    \Pr_{\mathbf{x} \in \mathbb{F}_q^n} [f(\mathbf{x})=a] - \Pr_{t \in \{0,1\}^s} [f(G(t))=a]
    \right| \leq \varepsilon.
    \]
    The \emph{seed length} of $G$ is $s$, and we say that $G$ is \emph{explicit} if for any $t \in \{0,1\}^s$, $G(t)$ can be computed in time $\operatorname{poly}(n,d,\log q, \log 1/\eps)$. 
\end{definition}
PRGs for low-degree polynomials have been extensively studied for more than three decades,
with the natural goal of \emph{minimizing the seed length}. Moreover, over the years it has
become apparent that constructing PRGs over small fields—of constant size, independent of
the degree~$d$—is more challenging than in the regime where the field size is allowed to be
polynomial in~$d$, which permits the use of deep results such as the Weil bound.

Already the case $d=1$, which corresponds to \emph{small-biased generators}, is extremely interesting and has found numerous applications in pseudorandomness and derandomization. In this setting, we have constructions with seed length that is optimal up to constant factors (see, e.g., \cite{NN90,ABNNR92,AGHP92,AMN98,BT13,T17,CC25,HMZ26}), typically over any field size (although the $\FF_2$ case is the most widely studied).

For arbitrary $d$, one can show a lower bound of $
s = \Omega\bigl(d\log(n/d) + \log(1/\eps) + \log q\bigr)
$ (see, e.g., \cite{ABK08}),
and the probabilistic method guarantees the existence of a construction achieving these parameters. From now on, we refer to this as an optimal seed (ignoring constant factors). Historically, obtaining explicit constructions developed along two main strands. 

\paragraph{PRGs over an arbitrary field.} Over any finite field, and in particular over the challenging case of $\FF_2$, a sequence of works \cite{LVW93,Vio07,BV10,Lov09,Vio09} culminated in Viola's celebrated explicit PRG with seed length $O(d\log n + d \cdot 2^{d}\log(q/\eps))$ \cite{Vio09}. The generators of \cite{BV10,Lov09,Vio09} follow the Bogdanov--Viola framework \cite{BV10}: to fool degree-$d$ polynomials, sum $\ell=\ell(d)$ independent copies of a small-bias generator. Viola \cite{Vio09} proved that $\ell(d)=d$ suffices. This dependence limits the construction to logarithmic degrees. Going beyond this regime was a longstanding challenge, closely tied to circuit lower bounds.

A recent breakthrough by Chattopadhyay, Hatami, Lee, Lovett, Tal, and Viola \cite{CHLLTV26} overcame this barrier over $\FF_2$, giving a PRG with seed length $\widetilde{O}(d^2\log^2 n)$ for constant error. Their construction uses new correlation bounds for low-degree polynomials, together with polarized random walks. Subsequent work by Cohen \cite{Cohen26} achieved seed length $O(d^4\log n)$, as well as $\widetilde{O}(d^3\log n)$, obtaining optimal dependence on $n$ at the expense of a larger polynomial dependence on $d$. These developments replace the exponential dependence on the degree with a polynomial one, while leaving open the challenge of achieving optimal seed length.

In yet another recent work, Yankovitz \cite{Yank26} constructed a \emph{hitting set generator} (HSG; see \cref{def:hsg}), a weaker object than a PRG, over the binary field with seed length $O(d \log n + d^2)$. Interestingly, this work adapts some of the ideas we introduced in \cite{CDG26} for constructing PRGs over large fields to the construction of HSGs over the binary field.

\paragraph{PRGs over large fields.} When $q \gg d$, stronger seed-length bounds are known, including optimal bounds under suitable field-size assumptions. Bogdanov \cite{Bog05} introduced a technique for constructing PRGs from hitting set generators. This approach is based on reducing the number of variables of the polynomial while preserving its factorization structure (i.e., preserving the irreducibility of its factors). Combined with subsequent improvements in HSG constructions due to Lu \cite{Lu12}, Cohen and Ta-Shma \cite{CT13}, and Guruswami and Xing \cite{GXHittingSets}, Bogdanov's technique yields a PRG with seed length $O(d^{4}\log n + \log q)$, provided that $q \ge d^{6}/\eps^2$.

More recently, Derksen and Viola~\cite{DV22} introduced a powerful new approach based on techniques from algebraic geometry and invariant theory.
For sufficiently large $q \ge (d^{4}n^{0.001})/\eps^2$, they achieve optimal seed length $O(d\log(dn)+\log q)$. For $q \ge (d\log n)^{4}/\eps^2$, they obtain a suboptimal seed length of $O(d \log n \cdot \log(d \log n)+\log q)$. In particular, they develop an algebraic analogue of the Bogdanov--Viola paradigm \cite{BV10} of summing generators for degree-one polynomials. We will discuss their approach, based on the preservation of indecomposability instead of irreducibility, in \cref{sec:overview}.

Subsequently, Dwivedi, Guo, and Volk \cite{dwivedi2024optimalpseudorandomgeneratorslowdegree}, using a quite different approach, based on Lecerf's technique \cite{Lec06,Lec07}, removed the dependence on $n$ in the field-size requirement needed to obtain an optimal-seed PRG, albeit with an exponential dependence on $d$. Specifically, they achieve seed length $O(d\log n + \log q)$ whenever $q \ge d2^{d}/\eps + d^{4}/\eps^2$ and the field characteristic is $\Omega(d^2)$.\footnote{Interestingly, if one only cares about fooling polynomials of \emph{prime degree}, \cite{dwivedi2024optimalpseudorandomgeneratorslowdegree} show that $q = \Omega(d^4/\eps^2)$ suffices.} We also discuss their technique, which combines ideas from \cite{DV22} with a derandomization approach inspired by \cite{Bog05}, in \cref{sec:overview}.

\subsection{Our Result}\label{sec:overvoew}
Our first result is an explicit PRG construction with optimal seed length over fields of size polynomial in $d$.

\begin{theorem}[see also \cref{thm:main}]\label{thm:main:intro}
For every $n,d \in \mathbb{N}$, a prime power $q$, and $\eps > 0$, satisfying  
$q = \Omega((d\log d)^4/\eps^2)$,
there exists an explicit PRG $G \colon \{ 0,1 \}^s \rightarrow \FF_q^n$ for $n$-variate polynomials of degree at most $d$ over $\FF_q$ with error $\eps$ and seed length $s = O(d\log n + \log q)$.
\end{theorem}

Our construction improves on \cite{DV22,dwivedi2024optimalpseudorandomgeneratorslowdegree} simultaneously: we achieve optimal seed length already for $q \ge \poly(d)$, replacing the exponential dependence on $d$ in \cite{dwivedi2024optimalpseudorandomgeneratorslowdegree} and eliminating the dependence on $n$ in \cite{DV22}. Moreover, our construction imposes no restriction on the characteristic. We note that our result is incomparable to \cite{CHLLTV26,Cohen26}, which handle the binary field but have larger seed lengths.

\medskip

We turn to present our second result. As discussed above, the study of PRGs for low-degree polynomials has developed along two largely separate lines: PRGs over fields of constant size, most notably the binary field, and PRGs over fields whose size is sufficiently large as a function of the degree~$d$ (and, in some constructions, also of the number of variables~$n$). These two lines of work rely on fundamentally different techniques. Our first result, \cref{thm:main:intro}, belongs to the second: we construct an optimal-seed PRG that works whenever $q$ is at least roughly $d^4$.

Our second result connects these two lines of work through a reduction from the small-field setting to the large-field setting. Given that these two lines of work have relied on fundamentally different techniques for many years, we find the existence of such a simple reduction somewhat surprising.

\begin{theorem}\label{thm:f2}
Fix a constant $\delta\in(0,1)$. Assume there exists an explicit PRG for degree-$d$ polynomials in $n$ variables over fields of characteristic two and size $q\ge d^\delta$, with constant error and optimal seed length $O(d\log n+\log q)$. Then, there exists an explicit PRG for degree-$d$ polynomials in $n$ variables over $\FF_2$, with the same error, and seed length
\[
O\!\left(d^{\frac{1}{1-\delta}}\log n\right).
\]
\end{theorem}

For example, constructing an optimal-seed-length PRG over fields of characteristic two and size $d^{3/4}$ would, via the reduction in \cref{thm:f2}, recover the seed length $O(d^4 \log n)$ achieved in \cite{Cohen26}. Extending such a construction to fields of size $\sqrt{d}$ would improve the seed length to $O(d^2 \log n)$.

The reduction is particularly simple: given a PRG over $\FF_q$, we apply the field trace $\Tr_{q\rightarrow 2}$ to each coordinate of its output. To see why this works, consider a degree-$d$ polynomial $f$ over $\FF_2$. Composing $f$ with the coordinatewise trace gives a polynomial over $\FF_q$ of degree at most $dq/2$, since the trace itself is a polynomial of degree $q/2$. Moreover, the trace maps a uniform element of $\FF_q$ to a uniform bit. Thus, fooling the composed polynomial over $\FF_q$ suffices to fool $f$ over $\FF_2$. The key is to choose $q$ so that the base PRG's field-size requirement remains satisfied after this increase in degree. The hypothesis of \cref{thm:f2} allows such a choice, yielding the stated seed length. The same argument works over any prime field $\FF_p$, using $\Tr_{q\rightarrow p}$; we focus on the binary field for simplicity. The formal proof appears in \cref{sec:f2}.

\subsection{Proof Overview}\label{sec:overview}

The final paragraph of the previous section already captures the idea behind the short and simple proof of \cref{thm:f2}. In this section, we give a high-level overview of the main ideas behind the more involved proof of \cref{thm:main:intro}. Our proof develops the approach of
\cite{DV22,dwivedi2024optimalpseudorandomgeneratorslowdegree}
through several new ideas. A central ingredient is a new pseudorandom
object, which we call a \emph{polynomial hitting set generator} (PHSG),
generalizing hitting set generators to allow polynomial evaluation
points. This allows us to achieve exponentially smaller error than
ordinary HSGs over the same field, at the cost of introducing a small
number of new variables. We use PHSGs to remove the exponential field-size requirement while retaining optimal seed length. We also give a simpler approach to preserving irreducibility, based
directly on Hensel lifting, that avoids the additional machinery of
Lecerf's technique and its large-characteristic requirement.
Together with a suitable change of variables, this allows us to carry
out the construction in arbitrary characteristic.

We begin by describing the framework introduced by Derksen and Viola~\cite{DV22}. Let $\mathbf{x} = (x_1,\ldots,x_n)$. A polynomial
\( f \in \mathbb{F}_q[\mathbf{x}] \) is called \emph{indecomposable} if it cannot be written as
\( f = g \circ h \), where \( h \in \mathbb{F}_q[\mathbf{x}] \) and \( g \in \mathbb{F}_q[t] \) is a
univariate polynomial with \( \deg(g) \geq 2 \). The notion of indecomposability is interesting partially because we understand the distribution of its image. This is formalized by the following lemma.
\begin{lemma}[\cite{DV22}, Lemma 12]\label{lemma:equidistribution in overview}
    There exists an absolute constant $c>0$ such that the following holds: Suppose $f\in\FF_q[\mathbf{x}]$ is indecomposable over $\FF_q$ of degree at most $d$. Then, $f(\mathbf{U}_{\FF_q^n})$ is $\varepsilon$-close to $\mathbf{U}_{\FF_q}$, where $\varepsilon=c \cdot d^2/\sqrt{q}$. 
\end{lemma}

The general approach in our work, following the works \cite{Bog05, GXHittingSets, DV22, dwivedi2024optimalpseudorandomgeneratorslowdegree}, is to restrict \(f\) to a carefully chosen subset of
\(\mathbb{F}_q^n\) while preserving some algebraic property, closely related to its output distribution.
In this work, we use the notion of indecomposability, although one may also work with irreducibility, as in \cite{Bog05, GXHittingSets}; this distinction plays a relatively minor role in our approach.
Specifically, let \( f \in \mathbb{F}_q[\mathbf{x}] \) be an \emph{indecomposable} polynomial of degree at most \( d \). Suppose we can find polynomials
\( p_1, \ldots, p_n \in \mathbb{F}_q[w_1, \ldots, w_\ell] \), with \( \ell \ll n \), such that for every such $f$, the
composed polynomial
\( f \circ (p_1, \ldots, p_n) \in \mathbb{F}_q[w_1, \ldots, w_\ell] \) is indecomposable. Let $\mathbf{w} = (w_1,\ldots,w_\ell)$, let
\( \mathbf{p} = (p_1, \ldots, p_n) \) be the restriction map, and let $\deg \mathbf{p} = \max_i \{\deg p_i\}$.
Now, \cref{lemma:equidistribution in overview} applies both to \( f \) and to its restriction via
\( f \circ \mathbf{p} \). In particular, the distribution of
\( f(\mathbf{U}_{\FF_q^n}) \) is $O(d^2 / \sqrt{q})$-close to uniform over \( \mathbb{F}_q \), and similarly the
distribution of \( f(\mathbf{p}(\mathbf{U}_{\FF_q^\ell})) \) is $O((d \cdot \deg \mathbf{p})^2 / \sqrt{q})$-close to \( \mathbf{U}_{\mathbb{F}_q} \).
This means that the function $G:\FF_q^\ell \rightarrow \FF_q^n$ defined by \[G(\mathbf{w}) = \mathbf{p}(\mathbf{w})\] is a PRG for indecomposable polynomials of degree at most $d$ with seed length $\ell \log q$, and error
\begin{equation}\label{eq:epsilon-in-PRG}
    \varepsilon = O ((d \cdot \deg \mathbf{p})^2 / \sqrt{q}).
\end{equation}

For an \emph{arbitrary} $n$-variate polynomial $f$ of degree at most $d$, we can always write
\( f = g \circ h \), where \( h \in \mathbb{F}_q[\mathbf{x}] \) is indecomposable and
\( g \in \mathbb{F}_q[t] \) is univariate. $h$ is an $n$-variate indecomposable polynomial of degree at most $d$, therefore the distributions $h(\mathbf{U}_{\FF_q^n})$ and $h( G(\mathbf{U}_{\FF_q^\ell}))$ are both close to $\mathbf{U}_{\FF_q}$.
Hence, the distributions \( f(\mathbf{U}_{\FF_q^n}) = g(h(\mathbf{U}_{\FF_q^n})) \) and $f(G(\mathbf{U}_{\FF_q^\ell})) = g(h(G(\mathbf{U}_{\FF_q^\ell})))$ are both close to $g(\mathbf{U}_{\FF_q})$. This shows that $G$ is actually a PRG for all $n$-variate polynomials of degree at most $d$.

The challenge is, of course, to find low-degree \(p_1,\ldots,p_n\) that preserve
indecomposability for all indecomposable $n$-variate polynomials $f$ of degree at most $d$. 

\subsubsection{The  Derksen-Viola restriction map}
A main contribution of \cite{DV22} is the explicit construction of such restriction polynomials \(p_1,\ldots,p_n\), which proceeds as follows. Let
\(M_1,\ldots,M_n\) be distinct monomials in \(m\) variables. Consider
\(\ell\) independent copies of these variables, and denote by
\(M_i^{[j]}\) the monomial \(M_i\) evaluated on the variables from the
\(j\)-th copy. Then define
\[
p_i \;=\; M_i^{[1]} + \cdots + M_i^{[\ell]} \, .
\]
Using tools from invariant theory, \cite{DV22} show that for a
suitable choice of parameters and monomials, the resulting substitution
\(\mathbf{p} = (p_1,\ldots,p_n)\) preserves indecomposability.
This allows them to construct a PRG with optimal seed length
\(O(d \log(dn) + \log q)\), assuming
\(q = \Omega(d^4 n^{0.001}/\varepsilon^2)\), or alternatively a PRG with seed length
\(O(d \log n \cdot \log(d \log n) + \log q)\), assuming
\(q = \Omega((d \log n)^4/\varepsilon^2)\).

\subsubsection{Constructing the restriction map via hitting set generators}\label{subsubsec:lecerf technique}
In contrast to the approach of \cite{DV22}, Bogdanov \cite{Bog05} and Dwivedi, Guo, and Volk
\cite{dwivedi2024optimalpseudorandomgeneratorslowdegree} adopt a
derandomization approach. Rather than constructing a single restriction map
\(\mathbf{p}\) that preserves indecomposability for all polynomials of
degree at most \(d\), they use a carefully designed distribution over
restriction maps \(\mathbf{p}\). This distribution has the property that for
every indecomposable polynomial \(f\) of degree at most \(d\), the
composition \(f \circ \mathbf{p}\) remains indecomposable with high
probability. In both works, the restriction polynomials are linear, and the distribution is designed using a pseudorandom object called a \textit{hitting set generator} (HSG).

\begin{definition}\label{def:hsg}
A function $H \colon T \rightarrow\FF_q^n$ is a 
\textit{hitting set generator} (HSG) with density $1-\delta$ for $n$-variate polynomials of degree at most $d$ over $\FF_q$ if for every such $f \neq 0$,
\[
    \Pr_{t \in T}[f(H(t)) \neq 0] \geq 1 - \delta.
\]
The \emph{seed length} of the HSG is $\log_2 |T|$, and we say that $H$ is \emph{explicit} if for any $t \in T$, $H(t)$ can be computed in time $\operatorname{poly}(n,d,\log q, \log 1/\delta)$. 
\end{definition}
We proceed by describing \cite{dwivedi2024optimalpseudorandomgeneratorslowdegree} construction in more detail.
Let
\[
p_i(x,y) =
\begin{cases}
    \beta_i x + \alpha_i y & 1 \leq i \leq n-1 \\
    y & i = n
\end{cases},
\]
where the vectors \(\mathbf{\alpha}, \mathbf{\beta} \in \mathbb{F}_q^n\) are sampled from a HSG for $(n-1)$-variate polynomials of degree at most $O(d)$ (where the implicit constant is absolute). This transformation can also be written as follows: Let $s_\mathbf{\alpha}$ be the ring automorphism of $\FF_q[\mathbf{x}]$ such that 
$$
s_\mathbf{\alpha}(f(\mathbf{x})) = f(x_1 + \alpha_1x_n,\ldots,x_{n-1}+\alpha_{n-1}x_n, x_n).
$$
Let $r_\mathbf{\beta} \colon \FF_q[\mathbf{x}]\rightarrow \FF_q[x,y]$ be the homomorphism such that $r_\mathbf{\beta}(f(\mathbf{x})) = f(\beta_1x,\ldots,\beta_{n-1}x,y)$. Then, for all $f \in \FF_q[\mathbf{x}]$ we have $$f \circ \mathbf{p} = (r_\mathbf{\beta} \circ s_\mathbf{\alpha})(f).$$
It is shown in \cite{dwivedi2024optimalpseudorandomgeneratorslowdegree} that for every
degree-$d$ \(f \in \mathbb{F}_q[\mathbf{x}]\), there exists a polynomial
\(B \in \mathbb{F}_q[x_1,\ldots,x_{n-1}]\) of degree at most \(d\) with the following
property. For any \(\boldsymbol{\alpha}\) satisfying \(B(\boldsymbol{\alpha}) \neq 0\),
the polynomial $$s_{\boldsymbol{\alpha}}(f) - t \in \mathbb{F}_q(t)[\mathbf{x}]$$
satisfies \emph{Hypothesis~(H)} (up to multiplication by an element of
\(\mathbb{F}_q^\times\)). Hypothesis~(H) is a condition required for applying
\emph{Hensel Lifting Lemma} which is used by \emph{Lecerf's technique}; see \cref{def:hypothesis H} for details. Thus, by picking $\mathbf{\alpha}$ using a HSG, it is promised that the polynomial $s_\alpha(f)-t$ satisfies Hypothesis~(H) with high probability. 

They proceed by showing, building on results of Lecerf~\cite{Lec06, Lec07}, that if
\(g \in \mathbb{F}_q[\mathbf{x}]\) is indecomposable and \(g - t\) satisfies
Hypothesis~(H), then the restriction \(r_{\boldsymbol{\beta}}(g)\) remains
indecomposable with high probability.

The proof proceeds roughly as follows: Let $\mathbb{K}$ = $\overline{\FF_q(t)}$ be the algebraic closure of $\FF_q(t)$. Let $g \in \mathbb{K}[\mathbf{x}]$ be a polynomial satisfying Hypothesis~(H). For such a polynomial $g$, assuming that $\text{char}\ \FF_q = \text{char}\, \mathbb{K} > d(d-1)$, Lecerf constructed a linear system of equations $D_g$ in $d$ variables, with the coefficients being $(n-1)$-variate polynomials in $\mathbb{K}[z_1,\ldots,z_{n-1}]$ of degree at most $2d-1$, with the following properties:
\begin{enumerate}
    \item The number of irreducible factors of $g$ equals the dimension of the space of solutions of $D_g$. \label{item:solutions space dimension}
    \item For $\beta \in \mathbb{K}^{n-1}$, let $D_{g}^\mathbf{\beta}$ be the system of equations $D_g$ after evaluating $(z_1,\ldots,z_{n-1})=\mathbf{\beta}$. The number of irreducible factors of $r_\mathbf{\beta}(g)$ equals the dimension of the space of solutions of $D_{g}^\mathbf{\beta}$.\label{item:solutions space dimesnsion after evaluation}
    \item There are sets $S_1\subseteq S_2 \subseteq \{0,1\}^{d}$ such that the space of solutions of $D_g$ is spanned by $S_1$, and the space of solutions of $D_g^\mathbf{\beta}$ is spanned by $S_2$. \label{item:solution space and irreducibility}
\end{enumerate}
The condition that a polynomial $f \in \FF_q[\mathbf{x}]$ is indecomposable is equivalent to $g = f-t \in \mathbb{K}[\mathbf{x}]$ being irreducible (see \cref{lemma:indecomposable equivalence to irreducibility}).
Hence, preserving indecomposability of $f$ is equivalent to preserving the space of solutions of $D_g$ in the reduction to $D_g^\mathbf{\beta}$.

Let $f \in \FF_q[\mathbf{x}]$ be an indecomposable polynomial satisfying Hypothesis~(H). Then $g = f-t \in \mathbb{K}[\mathbf{x}]$ is irreducible, and by \cref{item:solutions space dimension,item:solution space and irreducibility} the solution space of $D_g$ is spanned by a single vector $v \in \{0,1\}^d$. For every nonzero $v \neq w \in \{0,1\}^d$, there is an equation in $D_g$ such that $w$ does not satisfy, i.e., there are polynomials $P_1^w,\ldots,P_d^w \in \mathbb{K}[z_1,\ldots,z_{n-1}]$ of degree at most $2d-1$ such that $$Q_w := \sum_{i=1}^d w_i \cdot P_i^w \neq 0.$$ Notice that $Q_w \in \mathbb{K}[z_1,\ldots,z_{n-1}]$ is a polynomial of degree at most $2d-1$ as well. By \cref{item:solutions space dimesnsion after evaluation}, if $Q_w(\mathbf{\beta}) \neq 0$ for all nonzero $v \neq w \in \{0,1\}^d$, then the polynomial $r_\mathbf{\beta}(g)$ remains irreducible, which means that the polynomial $r_\mathbf{\beta}(f)$ is indecomposable.

This implies that for every indecomposable polynomial that satisfies Hypothesis~(H), there are $2^d-1$ polynomials in $\mathbb{K}[z_1,\ldots,z_{n-1}]$ of degree at most $2d-1$ such that if each one of them does not vanish at $\mathbf{\beta}$, then $r_\mathbf{\beta}(f)$ remains indecomposable. Hence, by picking $\mathbf{\beta}$ using an HSG for $n-1$-variate polynomials of degree at most $2d-1$ with small enough density, the probability that $r_\mathbf{\beta}(f)$ remains indecomposable is high.
Thus, using the union bound, the probability that the substitution $f \circ \mathbf{p}$ preserves indecomposability is at least
\(1 - (2^{d}-1)\delta\).

Any HSG for polynomials of degree at most $2d-1$ over \(\mathbb{F}_q\) has density at most
\(1-\delta\) with \(\delta = \Theta(d/q)\), as the density of zeroes of a degree-$d$ polynomial can be $\Omega(d/q)$. Consequently, the
analysis yields a meaningful guarantee only when $q = \Omega(d\,2^{d})$, which is where the exponential dependence in $d$ appears in the field size.

\subsubsection{Our construction}\label{subsubsec:our construction}
As explained above, the exponential dependence of the field size on \(d\) arises from applying a union bound over exponentially many polynomials. One of our main ideas for overcoming this requirement is to choose \(\boldsymbol{\beta}\) not as a vector over \(\mathbb{F}_q\), but as a vector of polynomials of bounded degree in a small number of variables. Let \(\mathcal A \subseteq \mathbb{F}_q[w_1,\ldots,w_\ell]\) be a finite set of such polynomials, and suppose each coordinate of \(\boldsymbol{\beta}\) is chosen independently and uniformly from \(\mathcal A\). By the Schwartz–Zippel lemma (over the field of rational functions), for every nonzero polynomial \(Q \in \mathbb{F}_q[x_1,\ldots,x_n]\) of degree at most \(d\),
\[
\Pr[Q(\boldsymbol{\beta})=0]\le \frac{d}{|\mathcal A|}.
\]
Since \(\mathcal A\) can be much larger than \(\mathbb{F}_q\), this allows us to apply the union bound without requiring \(q\) to be exponentially large in \(d\). This approach gives irreducibility with high probability over the field of rational functions in the new variables. The structure imposed by Hypothesis (H) also guarantees that the resulting polynomial is primitive, which allows us to apply Gauss’s lemma to obtain irreducibility over the desired polynomial ring.

This leads us to look for a derandomization of the choice of \(\boldsymbol{\beta}\), analogously to standard HSGs over \(\mathbb{F}_q\), while keeping the degrees and number of variables small. To achieve this, we introduce the notion of a \emph{polynomial hitting set generator} (PHSG).
\begin{restatable}{definition}{PHSGDefinition}
        Let $\FF$ be a finite field. A \emph{polynomial hitting set generator (PHSG)} with density $1-\delta$ for $n$-variate polynomials of degree at most $d$ over $\FF$ with $\ell$-variate polynomial evaluation points of degree at most $h$ is a map $$H \colon T \rightarrow (\FF^{\leq h}[w_1,\ldots,w_\ell])^n$$ from a finite set $T \neq \emptyset$ such that for every nonzero polynomial $f \in \FF[x_1,\ldots,x_n]$ of degree at most $d$,
    \[
        \Pr_{y \in T} [f(H(y)) = 0] \leq \delta.
    \]
    The quantity $\log |T|$ is called the \textit{seed length} of $H$.
\end{restatable}
In \cref{sec:PHSG}, we show that any HSG for polynomials over a field extension \(\mathbb{F}_{q^k}\) can be turned into a PHSG over \(\mathbb{F}_q\). Consequently, existing constructions of HSGs yield PHSGs. However, implementing this approach is quite involved, as we discuss in \cref{subsubsec:intro-PHSG-construction}. For now, assume that we have access to an explicit construction of PHSGs.

Using the PHSG in our construction and removing the restriction on
the characteristic require additional ideas. We now present our
construction and its analysis. We choose
the restriction polynomials to be
\[
p_i(x,y,w_1,\ldots,w_\ell)
=
b_i(w_1,\ldots,w_\ell)\cdot x+\alpha_i y+\alpha_i^2,
\qquad 1\le i\le n-1,
\]
and \(p_n(x,y,w_1,\ldots,w_\ell)=y\), where \(x,y,w_1,\ldots,w_\ell\)
are new variables. Here, \(\boldsymbol{\alpha}\in\mathbb{F}_q^{n-1}\)
is sampled from an HSG
\(H_1:T_1\to\mathbb{F}_q^{n-1}\) with density \(1-\delta_1\), and
\(\mathbf{b}(\mathbf{w}) = \mathbf{b}=(b_1,\ldots,b_{n-1})\) is sampled independently from a PHSG
\[
H_2:T_2\longrightarrow
\bigl(\mathbb{F}_q^{\le h}[w_1,\ldots,w_\ell]\bigr)^{n-1}
\]
with density \(1-\delta_2\).

We first explain the role of the quadratic terms \(\alpha_i^2\).
As in the preceding construction, we apply a change of variables
to obtain a polynomial satisfying Hypothesis (H). Here, we use the
automorphism
\[
s_{\boldsymbol{\alpha}}(f)
=
f(x_1+\alpha_1x_n+\alpha_1^2,\ldots,
  x_{n-1}+\alpha_{n-1}x_n+\alpha_{n-1}^2,x_n).
\]
While the change of variables used in
\cite{dwivedi2024optimalpseudorandomgeneratorslowdegree}
requires large characteristic to ensure Hypothesis (H), the addition
of the quadratic terms allows us to obtain this guarantee in arbitrary
characteristic. Specifically, in \cref{cor:HSG for Hyptothesis (H)} we show that, for every indecomposable
polynomial \(f\) of degree at most \(d\), there exists a nonzero
polynomial \(B\in\mathbb{F}_q[z_1,\ldots,z_{n-1}]\) of degree at most
\(3d-1\) such that \(B(\boldsymbol{\alpha})\ne 0\) guarantees that
\(s_{\boldsymbol{\alpha}}(f)-t\) satisfies Hypothesis (H), up to
multiplication by a nonzero constant. Thus, choosing \(H_1\) to hit
polynomials of degree at most \(3d-1\) ensures Hypothesis (H) with
probability at least \(1-\delta_1\).

We now fix \(\boldsymbol{\alpha}\) for which Hypothesis (H) holds,
and let \(g\) be the normalization of
\(s_{\boldsymbol{\alpha}}(f)-t\). Let
\(\mathbb{K}=\overline{\mathbb{F}_q(t)}\).
Recall that indecomposability of \(f\) over \(\mathbb{F}_q\) is
equivalent to irreducibility of \(f-t\) over \(\mathbb{K}\)
(\cref{lemma:indecomposable equivalence to irreducibility}). We therefore seek to preserve irreducibility under the substitution
\[
r_{\mathbf b}(g)
=
g(b_1(\mathbf w)x,\ldots,b_{n-1}(\mathbf w)x,y),
\]
viewed as a polynomial over \(\mathbb K(\mathbf w)\).
For \(\deg f=d\), Hypothesis (H) allows us to apply Hensel lifting
and obtain the factorization
\[
g(\mathbf x)
=
\prod_{i=1}^d
\bigl(x_n-\lambda_i(x_1,\ldots,x_{n-1})\bigr),
\]
where the \(\lambda_i\) are formal power series over \(\mathbb{K}\)
with pairwise distinct constant terms. In \cref{prop:HSG property for irreducibility}, we show that
any nontrivial factorization of $r_{\mathbf b}(g)$
corresponds to a partition of \([d]\) into two nonempty sets:
each factor is the
product of the corresponding linear factors of $r_{\mathbf b}(g)$.
For each such partition \(I,[d]\setminus I\), we find a nonzero
polynomial \(Q_I\in\mathbb{K}[z_1,\ldots,z_{n-1}]\) of degree at
most \(d\) such that \(Q_I(\mathbf b)=0\) whenever this
factorization occurs. Thus, as in the preceding approach, preserving
irreducibility reduces to avoiding the zero sets of low-degree polynomials. The number of polynomials remains exponential in \(d\), but we no longer require any restriction on the characteristic.

Fix \(\mathbf b\) that avoids all these zero sets,
then \(r_{\mathbf b}(g)\) is irreducible over \(\mathbb K(\mathbf w)\).
Hypothesis (H) ensures that \(r_{\mathbf b}(g)\) is monic in \(y\),
and hence primitive over \(\mathbb K[\mathbf w]\). Gauss's lemma
therefore gives irreducibility in \(\mathbb K[\mathbf w,x,y]\).
Since \(r_{\mathbf b}(g)\) is a nonzero constant multiple of
\(f\circ\mathbf p-t\), we conclude by
\cref{lemma:indecomposable equivalence to irreducibility}
that \(f\circ\mathbf p\) is indecomposable over \(\mathbb F_q\).

Altogether, by the PHSG guarantee and the union bound,
\(\mathbf{p}\) preserves indecomposability with probability at least
\[
    1 - \delta_1 - (2^{d-1}-1)\delta_2.
\]

In \cref{sec:PHSG}, we construct PHSGs with $\ell = h =:  \log k$, seed length
\(O(d\log n + k\log q)\) and density \(1-\delta\), where \(\delta \ge d/q^k\). Choosing $\delta_2 = O(d/q^k)$ with \(k = d/\log q + 1\), and $\delta_1 = O(d/q)$ gives
\[
    \Pr[f \circ \mathbf{p}\text{  is indecomposable}] = 1-O(d/q).
\]
This way we eliminate the requirement \(q=\Omega(d2^d)\).
This comes at the cost of increasing the number of variables in
\(\mathbf p\) from \(2\) to \(\ell+2\), and using a PHSG in
addition to the HSG. These increase the seed length only by a
constant factor: the PHSG instantiated with our parameters requires
seed length \(O(d\log n+k\log q)\), and sampling the new variables
uniformly from \(\mathbb F_q^{\ell+2}\) requires an additional
\((\ell+2)\log q\) random bits. As
\(k\log q=O(d+\log q)\) and \(\ell=\log k\), the resulting PRG
seed length is \(O(d\log n+\log q)\).

The requirement \(q=\Omega((d\log d)^4/\varepsilon^2)\) follows
from \cref{eq:epsilon-in-PRG}, since
$$
\deg\mathbf p\le h+1=\ell+1=O(\log d).
$$
To sum up, our pseudorandom generator
\(G:T_1\times T_2\times\FF_q\times\FF_q\times\FF_q^\ell
\rightarrow\FF_q^n\) is defined by
\begin{equation}
\begin{aligned}
G(r,s,u,v,\mathbf t)
=\bigl(&H_2(s)_1(\mathbf t)\cdot v
        +H_1(r)_1\cdot u+H_1(r)_1^2,\ldots,\\
       &H_2(s)_{n-1}(\mathbf t)\cdot v
        +H_1(r)_{n-1}\cdot u+H_1(r)_{n-1}^2,\ u\bigr),
\end{aligned}
\end{equation}
where \(\mathbf t=(t_1,\ldots,t_\ell)\) and \(u,v\in\FF_q\).

\subsubsection{The PHSG Construction}\label{subsubsec:intro-PHSG-construction}
In \cref{sec:PHSG}, we show that any HSG for polynomials over a field
extension \(\mathbb{F}_{q^k}\) can be turned into a PHSG over \(\mathbb{F}_q\). Hence, the construction of explicit PHSG is reduced to the construction of explicit HSG over a field extension.
An obstacle in that approach is that explicitly constructing an HSG over a field extension
requires an explicit representation of \(\mathbb{F}_{q^k}\), which in turn necessitates
the construction of irreducible polynomials over \(\mathbb{F}_q\). Unfortunately, the best-known deterministic algorithms which find an irreducible polynomial over a finite field are too costly.\footnote{Very recently, OpenAI claimed to have found such a deterministic
algorithm. This claim was made after the first version of this
work \cite{CDG26} appeared. We nevertheless believe that our approach
to avoiding the need for such an algorithm is insightful in its own right.}

To overcome this issue, we adopt a pseudorandomness approach. We design an efficient randomized algorithm that uses only a small number of random bits, which finds an irreducible polynomial of certain degree over a finite field with high probability. This algorithm is incorporated into our PRG, with the required randomness supplied by the PRG seed. It succeeds with high probability over the choice of seed, and we manage to translate its failure probability into only a small additional error in the resulting PRG.

To implement the algorithm, we choose the required irreducible polynomials at random and use a \textit{sampler} to amplify the success probability while incurring only a small additional randomness cost. Fortunately, this leads to only a constant-factor increase in the overall seed length of the resulting PRG.

%% file: preliminaries.tex
\section{Preliminaries}

We denote by $\mathbb{N}$ the set of nonnegative integers, including 0.
Throughout the paper, boldface letters denote vectors; for example,
$\mathbf{x} = (x_1,\ldots,x_n)$.
For a multi-index $\mathbf{i} = (i_1,\ldots,i_n) \in \mathbb{N}^n$, we use the standard notation
$
\mathbf{x}^{\mathbf{i}} := x_1^{i_1}\cdots x_n^{i_n}.
$

\paragraph{Resultants.}
Let $R$ be a commutative ring and let $f(x) = \sum_{i=0}^{d_1} a_i x^i$ and $g(x) = \sum_{i=0}^{d_2} b_i x^i$ be polynomials in $R[x]$, with $a_{d_1} \neq 0$ and $b_{d_2} \neq 0$. Suppose that $d_1 + d_2 \geq 1$.

\begin{definition}[Sylvester matrix]
The \emph{Sylvester matrix} $\Syl(f,g)$ of $f$ and $g$ is the following $(d_1+d_2)\times(d_1+d_2)$ matrix defined over $R$:
\[
\begin{pmatrix}
	a_{0} & & & &  b_{0} & & & & \\
	a_{1} & a_{0} & &  & b_{1} & b_{0} & & \\
	a_{2} & a_{1} & \ddots & & b_{2} & b_{1} & \ddots & & \\
	\vdots & & \ddots & a_{0} & \vdots & & \ddots & & b_{0} \\
	& \vdots & & a_{1} & b_{d_2} & \vdots & &  & b_{1} \\
	a_{d_1} & & & &  & b_{d_2}  & & & \\
	& a_{d_1} & & \vdots & &  & & \vdots & \\
	& & \ddots & & & & \ddots & & \\
	& & & a_{d_1} & & & &  & b_{d_2}
\end{pmatrix}_{\raisebox{2pt}{.}}
\]
\end{definition}

\begin{definition}[resultant]
The \emph{resultant} of $f$ and $g$, denoted $\Res{f,g}$, is defined as
\[
\Res{f,g} := \det\bigl(\Syl(f,g)\bigr) \in R.
\]
\end{definition}

The resultant satisfies the following property.

\begin{lemma}\label{lemma:resultant-vanishing}
Let $R$ be an integral domain with field of fractions $\mathbb{L}$.
Then $\Res{f,g}=0$ if and only if $f$ and $g$ have a common root in $\overline{\mathbb{L}}$.
\end{lemma}

\begin{definition}[formal power series]
Let $\FF$ be a field and let $\mathbf{x} = (x_1,\ldots,x_n)$ be indeterminates.
The ring of \emph{formal power series} $\FF[[\mathbf{x}]]$ consists of all infinite sums
\[
\sum_{\mathbf{i} \in \mathbb{N}^n} a_\mathbf{i} \mathbf{x}^\mathbf{i} \qquad (a_\mathbf{i} \in \FF),
\]
with addition and multiplication defined formally.
\end{definition}

\subsection{Prior Results}

We will use the optimal HSG (over large fields) of Guruswami and Xing.

\begin{theorem}[\cite{GXHittingSets}, Theorem 5.1]\label{thm:HSG construction}
    There exists an absolute constant $c$ such that 
    for any $n,d,q,\delta$, for which $q \geq c\cdot d/\delta$,
    there exists an efficiently computable HSG for $n$-variate polynomials of degree at most $d$ over $\FF_q$ with density $1-\delta$ and seed length $O(d\log n + \log (1/\delta))$.
\end{theorem}
We will also make use of the following lemma.
\begin{lemma}[\cite{dwivedi2024optimalpseudorandomgeneratorslowdegree}, Fact 2.4]\label{lemma:HSG remains in extension}
    Let $H$ be an HSG with density $1-\delta$ for polynomials of degree at most $d$ over a field $\FF$, and let $\KK$ be an extension of $\FF$. Then $H$ is also an HSG with density $1-\delta$ for polynomials of degree at most $d$ over $\KK$.
\end{lemma}

\subsection{Indecomposable Polynomials}
The analysis of our construction uses the notion of \textit{indecomposable} polynomials and some of their properties.
    \begin{definition}
    A non-constant polynomial $f\in\FF[\mathbf{x}]$ is said to be \emph{decomposable} over $\FF$ if there exist $h\in\FF[\mathbf{x}]$ and a univariate polynomial $g\in\FF[z]$ such that $\deg(g)\geq 2$ and $f=g \circ h$.
    Otherwise, $f$ is said to be \emph{indecomposable}.
\end{definition}
We will use the following equivalences.
\begin{lemma}[\cite{Bodin_2009}, Theorem 4.2]\label{lemma:indecomposable over algebraic closure}
    A polynomial $f\in\FF_q[\mathbf{x}]$ is indecomposable over $\FF_q$ if and only if it is indecomposable over $\overline{\FF_q}$.
\end{lemma}
\begin{lemma}[\cite{cheze:hal-00629517}, Lemma 7]\label{lemma:indecomposable equivalence to irreducibility}
    Let $f\in\FF[\mathbf{x}]$ be a non-constant polynomial over a field $\FF$.
        Then $f$ is indecomposable over $\overline{\FF}$ iff $f - t$ is irreducible over $\overline{\FF(t)}$, where $t$ is a new formal variable.
\end{lemma}
Using the Weil bound, Derksen and Viola have shown in \cite{DV22} that indecomposable polynomials are approximately equidistributed over large enough fields.
\begin{lemma}[\cite{DV22}, Lemma 12]\label{thm:indecomposable implies equidistribution}
    There exists an absolute constant $c>0$ such that the following holds: Suppose $f\in\FF_q[\mathbf{x}]$ is indecomposable over $\FF_q$ of degree at most $d$. Then $f(\mathbf{U}_{\FF_q^n})$ is $\varepsilon$-close to $\mathbf{U}_{\FF_q}$, where $\varepsilon=c \cdot d^2/\sqrt{q}$.
\end{lemma}

\subsection{Gauss's Lemma}
Gauss’s lemma provides a fundamental link between irreducibility over an integral domain and irreducibility over its field of fractions.
\begin{definition}
    Let $R$ be a unique factorization domain (UFD). Let $f(x) = \sum_{i=0}^d c_ix^i \in R[x]$. The \textit{content} of $f$, denoted by $c(f)$, is the greatest common divisor of $c_0,\ldots , c_d$; the content is well defined up to invertible elements in $R$. $f$ is called primitive if $c(f) = 1$.
\end{definition}
\begin{lemma}[Gauss]\label{lemma:Gauss}
    Let $R$ be a UFD, and let $\mathbb{L}$ be its field of fractions. A non-constant polynomial $f \in R[x]$ is irreducible in $R[x]$ if and only if it is both irreducible in $\mathbb{L}[x]$ and primitive in $R[x]$.
\end{lemma}
The following corollary follows by \cref{lemma:Gauss}, together with the fact that the polynomial ring $\FF[x_1,\ldots,x_n]$ is a UFD for every $n$ and field $\FF$.
\begin{corollary}\label{lemma:multivariate Guass}
    Let $\FF$ be a field and let $R = \FF[x_1,\ldots, x_n]$, $\mathbb{L} = \FF(x_1, \ldots , x_n)$. Then, a multivariate polynomial $f \in \FF[\mathbf{x}, y]$ such that $f \not \in  \FF[\mathbf{x}]$ is irreducible in $\FF[\mathbf{x}, y]$ if and only if it is irreducible in $\FF(\mathbf{x})[y]$ and $c(f) = 1$.
\end{corollary}

%% file: Evaluations_that_Retain_Irreducibility.tex
\section{Evaluations that Preserve Irreducibility}
Let \(\mathbb L/\mathbb K\) be an extension of algebraically closed
fields, and let \(f\in\mathbb K[x_1,\ldots,x_n,y]\) be irreducible.
In this section, we seek algebraic conditions on
\(\mathbf a\in\mathbb L^n\) under which
\(f(a_1x,\ldots,a_nx,y)\) remains irreducible over \(\mathbb L\).
Our argument uses Hensel lifting to write \(f\) as a product of
linear factors whose roots are power series over $\mathbb K$.
To apply Hensel lifting in this way, we need the polynomial to
satisfy Hypothesis (H). We therefore first explain how to obtain this hypothesis in
arbitrary characteristic by modifying the change of variables of
\cite{dwivedi2024optimalpseudorandomgeneratorslowdegree}.
As in that work, we then express the required conditions as the
nonvanishing of low-degree polynomials. We give an alternative argument which is simpler,
avoids Lecerf's technique, and works in arbitrary characteristic.
These conditions will allow us to construct restriction maps
that preserve indecomposability.

\subsection{Hypothesis (H)}
In order to apply Hensel lifting, we need to assume that our polynomial meets two standard conditions. Lecerf (\cite{Lec06,Lec07}) called these conditions Hypothesis~(H).
\begin{definition}[Hypothesis (H)]\label{def:hypothesis H}
    Let $f\in\FF[x_1,\dots,x_n,y]=\FF[\mathbf{x}, y]$ be a non-constant polynomial. We say $f$ satisfies \emph{Hypothesis~(H)} if
    \begin{enumerate}
    \item $f$ is monic in $y$ and $\deg_y(f)=\deg(f)$,\label{item:monic}
    \item $\Res{f(\mathbf{0},y),\frac{\partial f}{\partial y}(\mathbf{0},y)}\neq 0$.\label{item:resultant}
    \end{enumerate}
\end{definition}
Under a large-characteristic assumption,
\cite{dwivedi2024optimalpseudorandomgeneratorslowdegree} showed
that a suitable affine change of variables makes \(f-t\) satisfy
Hypothesis~(H),
provided a certain algebraic condition holds.
We add a quadratic term to the change of variables,
which allows us to obtain this guarantee in arbitrary characteristic.
\begin{definition}
    For $\mathbf{a}=(a_1,\dots,a_n)\in\FF^n$, let $s_{\mathbf{a}}$ be the $\FF$-linear automorphism of $\FF[\mathbf{x}, y]$ that fixes $y$ and sends $x_i$ to $x_i+a_i y+a_i^2$.
\end{definition}
\begin{lemma}\label{cor:HSG for Hyptothesis (H)}
Assume that $f\in \FF[\mathbf{x}, y]$ is a polynomial of degree $d\geq 1$ with at least one nonzero partial derivative. Then there exists a nonzero polynomial $B\in \FF[\mathbf{x}]$ of degree at most $3d-1$ such that for every $\mathbf{a}\in\FF^n$ satisfying $B(\mathbf{a})\neq 0$, we have $s_{\mathbf{a}}(f)-t = c\cdot g$, where $c\in\FF^\times$ and $g\in\FF(t)[\mathbf{x}, y]$ is a degree-$d$ polynomial satisfying Hypothesis~(H).
\end{lemma}

\begin{proof}
We construct \(B\) as a product of two polynomials, so that
\(B(\mathbf a)\ne 0\) implies that neither vanishes at \(\mathbf a\).
The first will ensure that \(s_{\mathbf a}(f)\) has degree $d$ in \(y\),
and the second that \(s_{\mathbf a}(f)(\mathbf 0,y)\)
has a nonzero derivative.

Let \(f_d\) be the homogeneous part of \(f\) of degree \(d\).
Since \(f_d\) is nonzero and homogeneous, \(f_d(\mathbf z,1)\)
is a nonzero polynomial of degree at most \(d\).
The coefficient of \(y^d\) in \(s_{\mathbf a}(f)\) is
\(f_d(\mathbf a,1)\). Thus, whenever \(f_d(\mathbf a,1)\ne 0\),
the polynomial \((s_{\mathbf a}(f) - t)/f_d(\mathbf a,1)\) satisfies \cref{item:monic}.

To obtain the second condition, introduce new variables
\(\mathbf z=(z_1,\ldots,z_n)\) and define
$$
h(\mathbf z,y)
=
f(z_1^2+z_1y,\ldots,z_n^2+z_ny,y).
$$
Thus, \(h(\mathbf a,y)=s_{\mathbf a}(f)(\mathbf 0,y)\).
We claim that \(\partial h/\partial y\ne 0\).

Let \(\varphi:\mathbb F[\mathbf x,y]\to\mathbb F[\mathbf z,y]\)
be the injective homomorphism that fixes \(y\)
and sends \(x_i\) to \(z_i^2+z_i y\).
Suppose that
\(\partial h/\partial y=0\). By the chain rule,
\begin{equation}\label{eq:chain-rule}
\sum_{j=1}^n z_j\,
\varphi\left(\frac{\partial f}{\partial x_j}\right)
+
\varphi\left(\frac{\partial f}{\partial y}\right)
=0.
\end{equation}
Fix \(i\in[n]\). Notice that the substitution
\(z_i\mapsto -z_i-y\) fixes the polynomial \(z_i^2+z_i y\).
Thus, applying it to \cref{eq:chain-rule} changes only the
coefficient \(z_i\) to \(-z_i-y\). Subtracting the resulting
equation from the original one gives
$$
(2z_i+y)\,
\varphi\left(\frac{\partial f}{\partial x_i}\right)=0.
$$
This shows that $\frac{\partial f}{\partial x_i} = 0$ for all $i \in [n]$. Substituting this back into \cref{eq:chain-rule} shows that $\frac{\partial f}{\partial y} = 0$ as well, contradicting the assumption that $f$ has at least one nonzero partial derivative.

Choose some nonzero coefficient \(C(\mathbf z)\) of
\(\partial h/\partial y\), viewed as a polynomial in \(y\).
If \(C(\mathbf a)\ne 0\), then
$$
0 \neq \frac{\partial h}{\partial y}(\mathbf a,y)
=
\frac{\partial}{\partial y}
\bigl(s_{\mathbf a}(f)(\mathbf 0,y)\bigr) \in \FF[y] \subseteq \FF (t)[y].
$$
Every root \(\rho\) of this polynomial is algebraic over
\(\mathbb F\). Hence \(s_{\mathbf a}(f)(\mathbf 0,\rho)\)
is also algebraic over \(\mathbb F\), and cannot equal \(t\).
Thus, \(s_{\mathbf a}(f)(\mathbf 0,y)-t\) and its derivative
have no common root, and hence $s_\mathbf{a}(f) - t$ satisfies \cref{item:resultant}.
Moreover, \(\deg h\le 2d\) implies \(\deg C\le 2d-1\).

Finally, set \(B(\mathbf x)=f_d(\mathbf x,1)C(\mathbf x)\).
Then \(B\ne 0\) and \(\deg B\le 3d-1\).
For every \(\mathbf a\) with \(B(\mathbf a)\ne 0\), let
\(c=f_d(\mathbf a,1)\). Above we proved that
\(g=c^{-1}(s_{\mathbf a}(f)-t)\) satisfies both conditions
of Hypothesis~(H), which completes the proof.
\end{proof}

\subsection{Hensel Lifting}
In this subsection, we assume that we are given an irreducible
polynomial satisfying Hypothesis~(H). Our main goal is
\cref{prop:HSG property for irreducibility}, which gives algebraic
conditions that guarantee that an evaluation point is
\emph{Bertinian good}. We begin with the definition.
\begin{definition}[Bertinian good points]
    Let $f\in\KK[\mathbf{x}, y]$ be a non-constant polynomial satisfying Hypothesis~(H), and let $\mathbb{L} / \mathbb{K}$ be a field extension. We say $\mathbf{a}=(a_1,\dots,a_n)\in\mathbb{L}^n$ is a \emph{Bertinian good point} for $f$ if for every irreducible factor $\tilde{f}$ of $f$ over $\mathbb{L}$, the bivariate polynomial $\tilde{f}_{\mathbf{a}}(x, y)=\tilde{f}(a_1 x,\dots,a_n x, y)$ is also irreducible over $\mathbb{L}$.
    Otherwise, $\mathbf{a}$ is called a \emph{Bertinian bad point}.
\end{definition}

The following is a special case of Hensel's Lifting Lemma (see, e.g, \cite{dwivedi2024optimalpseudorandomgeneratorslowdegree}, Lemma 2.9).

\begin{lemma}\label{lemma:Hensel lifting}
    Let $f\in\KK[x_1,\dots,x_n,y]=\KK[\mathbf{x},y]$ be a nonzero polynomial. Suppose $\bar{\lambda}\in \KK$ is a simple root of $f(\mathbf{0},y)\in\KK[y]$.
Then there exists a unique $\lambda\in\KK[[\mathbf{x}]]$ such that
\begin{enumerate}
    \item $f(\mathbf{x},\lambda)=0$, i.e., $\lambda$ is a root of $f$ as a univariate polynomial in $y$ over $\KK[\mathbf{x}]$, and
    \item $\lambda(\mathbf{0})=\bar{\lambda}$.
\end{enumerate}
\end{lemma}

Let $\mathbb{L} / \KK$ be an extension of algebraically closed fields. Let $f \in \KK[\mathbf{x}, y]$ be a polynomial of degree $d \geq 1$ satisfying Hypothesis (H).
Define $\bar{f}:=f(\mathbf{0}, y)\in \KK[y]$.
As $\KK$ is algebraically closed and $\Res{f(\mathbf{0},y),\frac{\partial f}{\partial y}(\mathbf{0},y)}\neq 0$, the univariate polynomial $\bar{f}$ factorizes into distinct linear factors
\[
\bar{f}(y)=\prod_{i=1}^d  (y -\bar{\lambda}_i) 
\]
where $\bar{\lambda}_i\in\KK$ for all $i\in [d]$.
By \cref{lemma:Hensel lifting}, the above factorization of $\bar{f}$ over $\KK$ lifts to a factorization of $f$ into distinct linear factors
\[
f(\mathbf{x}, y)=\prod_{i=1}^d (y-\lambda_i(\mathbf{x})),
\] 
where $\lambda_i\in \KK[[\mathbf{x}]]$ and $\lambda_i(\mathbf{0})=\bar{\lambda}_i$ for all $i\in [d]$.

We now introduce new variables $\mathbf{z}=(z_1,\dots,z_n)$ and $x$, and define
\[
    g:=f(z_1 x ,\dots, z_n x, y)\in \KK[\mathbf{z},x,y].
\]
Then, $g$ factorizes into linear factors
\[
g(\mathbf{z}, x, y)=\prod_{i=1}^d (y-\lambda_i(z_1 x, \dots, z_n x))
\]
where each $$\lambda_i(z_1 x, \dots, z_n x) \in \KK[[\mathbf{z}]][[x]].$$
For $i\in [d]$, let 
\[
    g_i := y-\lambda_i(z_1 x, \dots, z_n x) \in \KK[[\mathbf{z}]][[x]][y].
\]
For $h\in A[[x]][y]$ over a commutative ring $A$ and $(j,k)\in\mathbb{N}^2$, denote by $\coeff{h, x^j y^k}\in A$ the coefficient of $x^j y^k$ in $h$.

\begin{definition}[Truncated products of Hensel's Factors]\label{defn:system}
For $I\subseteq[d]$, define
\[
    P_I=\sum_{\substack{j,k\in\mathbb{N}\\j+k\leq |I|}}\coeff{\prod_{i\in I}g_i,x^jy^k}\,x^jy^k.
\]
\end{definition}

The following lemma is straightforward from the definitions.

\begin{lemma}\label{lem:degree-bound}
For every $I\subseteq[d]$, $P_I\in\KK[\mathbf{z},x,y]$ is monic in $y$ of degree $|I|$, and has total degree at most $|I|$ in $(x,y)$. For $(j,k)\in\mathbb{N}^2$, $\coeff{P_I,x^jy^k}$ is either zero or homogeneous of degree $j$ in $\mathbf{z}$.
\end{lemma}

\begin{lemma}\label{prop:HSG property for irreducibility}
Let $\KK$ be an algebraically closed field, and let $\mathbb{L}/\KK$ be an extension which is algebraically closed. Let $f\in\KK[\mathbf{x},y]$ be an irreducible polynomial of degree $d\geq 1$ satisfying Hypothesis~(H). Let $m=2^{d-1}-1$. Then there exist nonzero polynomials $Q_1,\dots,Q_m\in\KK[\mathbf{z}]=\KK[z_1,\dots,z_n]$ of degree at most $d$ such that for every Bertinian bad point $\mathbf{a}\in\mathbb{L}^n$ for $f$ over $\mathbb{L}$, at least one polynomial $Q_i$ vanishes at $\mathbf{a}$.
\end{lemma}
\begin{proof}
For every proper subset $I\subsetneq[d]$ containing $1$, let
\[
    R_I=g-P_IP_{[d]\setminus I}.
\]
We claim that $R_I\neq 0$. Otherwise, setting $x=1$ would give
\[
    f(\mathbf{z},y)=P_I(\mathbf{z},1,y)P_{[d]\setminus I}(\mathbf{z},1,y),
\]
a nontrivial factorization over $\KK$. This contradicts the irreducibility of $f$.

The polynomial $R_I$ has total degree at most $d$ in $(x,y)$. By \cref{lem:degree-bound} and the definition of $g$, every coefficient of $x^jy^k$ in $R_I$ is either zero or homogeneous of degree $j$ in $\mathbf{z}$. Choose one nonzero coefficient and denote it by $Q_I$. Then $\deg Q_I\leq d$.

We now show that every Bertinian bad point is a zero of at least one of these polynomials. Let $\mathbf{a}\in\mathbb{L}^n$ be Bertinian bad, and write
\[
    g(\mathbf{a},x,y)=g_1(x,y)g_2(x,y),
\]
where $g_1,g_2\in\mathbb{L}[x,y]$ are nonconstant. We may assume that both $g_1$ and $g_2$ are monic in $y$.
Setting $x=0$, we obtain
\[
    g_1(0,y)g_2(0,y)=f(\mathbf{0},y)=\prod_{i=1}^{d}(y-\bar{\lambda}_i).
\]
The roots $\bar{\lambda}_1,\ldots,\bar{\lambda}_d$ are pairwise distinct. Thus each $\bar{\lambda}_i$ is a root of exactly one of $g_1(0,y)$ and $g_2(0,y)$. Define
\[
    I=\{i\in[d]:g_1(0,\bar{\lambda}_i)=0\}.
\]
Consequently,
\[
    g_1(0,y)=\prod_{i\in I}(y-\bar{\lambda}_i),\qquad g_2(0,y)=\prod_{i\notin I}(y-\bar{\lambda}_i),
\]
and $I$ is nonempty and proper. Interchanging $g_1$ and $g_2$, if necessary, we may assume $1\in I$.
For each $i\in I$, the simple root $\bar{\lambda}_i$ of $g_1(0,y)$ lifts via Hensel lifting to a root of $g_1(x,y)$, and hence also to a root of $g(\mathbf{a},x,y)$. By uniqueness of Hensel lifting, this root is
\[
    \lambda_i(a_1x,\ldots,a_nx).
\]
Since $g_1$ is monic of $y$-degree $|I|$, it follows that
\[
    g_1(x,y)=\prod_{i\in I}\bigl(y-\lambda_i(a_1x,\ldots,a_nx)\bigr).
\]
Similarly,
\[
    g_2(x,y)=\prod_{i\notin I}\bigl(y-\lambda_i(a_1x,\ldots,a_nx)\bigr).
\]
Since the degrees of $g_1$ and $g_2$ are bounded by $|I|$ and $d-|I|$, respectively, we conclude that
\[
    P_I(\mathbf{a},x,y)=g_1(x,y),\qquad P_{[d]\setminus I}(\mathbf{a},x,y)=g_2(x,y).
\]
Hence
\[
    R_I(\mathbf{a},x,y)=g(\mathbf{a},x,y)-g_1(x,y)g_2(x,y)=0.
\]
In particular, its chosen coefficient satisfies $Q_I(\mathbf{a})=0$.
There are exactly $2^{d-1}-1$ proper subsets of $[d]$ containing $1$. This completes the proof.

\end{proof}

%% file: HSG_with_polynomial_evaluations.tex
\section{HSG with Polynomial Evaluation Points}\label{sec:PHSG}
In this section we construct our \emph{polynomial hitting set generators}, and specifically, show that \emph{any} hitting set generator for polynomials over a field
extension \(\mathbb{F}_{q^k}\) can be simply turned into a polynomial hitting
set generator (PHSG) over \(\mathbb{F}_q\). Then we use samplers to construct a field extension of $\FF_q$ efficiently using only a small amount of random bits.

The main advantage of PHSGs over HSGs is their ability to achieve much higher density,
namely $1-\Theta(d/q^k)$ rather than $1-\Theta(d/q)$.

We first recall the definition of a PHSG.
\PHSGDefinition
Our approach for constructing PHSGs is using a HSG for a field extension $$\mathbb{E} = \FF[w_1,\ldots,w_\ell] / P,$$ where $P \lhd \FF[w_1,\ldots , w_\ell]$ is a maximal ideal. Assume that in every equivalence class $g + P \in \mathbb{E}$ there exists an element $g' \in \FF[w_1,\ldots,w_\ell]$ of total degree at most $h$. Let $\{g_1 + P,\ldots,g_k + P\} \subseteq \mathbb{E}$ be a basis of $\mathbb{E}/\mathbb{F}$, and for all $1 \leq i \leq k$ let $g_i' \in \FF[w_1,\ldots,w_\ell]$ be an element in the equivalence class $g_i + P$ of total degree at most $h$.
Let
\[
    \varphi \colon \mathbb{E} \longrightarrow \FF^{\leq h}[w_1, \ldots , w_\ell]
\]
be the $\FF$-linear map defined uniquely by $\varphi(g_i + P) = g_i'$ for all $1 \leq i \leq k$.
Let $\pi \colon \FF[w_1, \ldots , w_\ell] \rightarrow \mathbb{E}$ be the ring homomorphism $a \mapsto a \mod P$. For all $a \in \mathbb{E}$ we have
\[
    \pi(\varphi(a)) = a.
\]
With this notation in place, we show that an HSG over $\mathbb{E}$ can be regarded as a PHSG. Whenever we apply $\varphi$ or $\pi$ to a vector, we apply the map coordinate-wise.
\begin{claim}\label{claim:HSG for univariate polynomials}
    Let $\widehat{H} \colon S \rightarrow \mathbb{E}^n$ be a HSG with density $1-\delta$ for $n$-variate polynomials over $\mathbb{E}$ of degree at most $d$. Then, $$H = \varphi \circ \widehat{H} \colon S \rightarrow (\FF^{\leq h}[w_1,\ldots,w_\ell])^n$$ is a PHSG with density $1-\delta$.
\end{claim}
\begin{proof}
    Let $f \in \mathbb{F}[x_1,\ldots,x_n]$.
    Let $s \in S$ such that $f(\widehat{H}(s)) \neq 0$. Since $\pi$ is a ring homomorphism, we obtain
    \[
        0 \neq f(\widehat{H}(s)) = f(\pi(H(s))) = \pi(f(H(s))),
    \]
    and in particular $f(H(s)) \neq 0$.
\end{proof}

\begin{remark}
    We note that $\mathbb{E}$ can equivalently be represented using a single irreducible polynomial $p \in \FF[w]$ of degree $k$, by setting $\mathbb{E} = \FF[w]/(p(w))$. Alternatively, we will use several irreducible polynomials of degree 2 in order to optimize the requirement on the field
    size \(q\). As we will see in \cref{sec:our PRG construction}, the field-size requirement depends on the degrees of the
    polynomial representations of elements of \(\mathbb{E}\). Constructing \(\mathbb{E}\) via
    a single degree-\(k\) irreducible polynomial would lead to a requirement of
    \(q \ge d^8/\varepsilon^2\), whereas our construction yields the improved bound
    \(q \ge (d\log d)^4/\varepsilon^2\).
\end{remark}
\subsection{Sampling Irreducible Polynomials}
Toward constructing PHSGs, we require an efficient method for building the extension
field \(\mathbb{E}\); in particular, we must construct irreducible polynomials over
\(\mathbb{F}_q\). At present, no unconditional deterministic algorithm is known for
constructing an irreducible polynomial of degree \(d\) over a field of size \(q\) in time
\(\mathrm{poly}(\log q, d)\).
Relevant progress includes the work of Shoup~\cite{Shoup1990}, who gave an algorithm
running in time \(O((\sqrt{p} + \log^2 q)\, d^{4})\), where \(p = \mathrm{char}(\mathbb{F}_q)\),
and the work of Adleman and Lenstra~\cite{AdlemanLenstra1986}, who provided a deterministic
\(\mathrm{poly}(\log q, d)\)-time algorithm assuming the Extended Riemann Hypothesis. We overcome this difficulty by using randomness, which we can consider as part of the construction's seed. 

Consider sampling a random degree-$a$ monic polynomial, which requires $O(a \log q)$ random bits. It is well known that with probability $\Theta(1/a)$, it will be irreducible. We wish to amplify this probability to $1-\delta$ for an arbitrary $\delta > 0$. To do it randomness-efficiently, we use randomness samplers.

\begin{definition}[sampler]
    Let $n, m, t$ be some positive integers, and let $
    \varepsilon, \delta > 0$. A function $S \colon\{0,1\}^n \rightarrow (\{0,1\}^m)^t$ is a $(\delta, \varepsilon)$ averaging sampler with $t$ samples using $n$ random bits if for every function $f \colon \{0,1\}^m \rightarrow [0,1]$ we have
    \[
        \Pr_{(Z_1,\ldots,Z_t) \sim S(U_n)}\left[ \left |\frac{1}{t} \sum_i f(Z_i) - \mathbb{E}[f] \right |  \leq \varepsilon \right ] \geq 1 - \delta.
    \]
\end{definition}
We use the recent nearly-optimal sampler of \cite{samplers_construction} (although some earlier constructions, such as \cite{Zuc97} and \cite{RVW00}, would work just as well for our use up to constant factors).
\begin{theorem}[\cite{samplers_construction}]\label{thm:sampler construction}
    For every constant $\beta > 0$, and for every $0 < \delta \leq \varepsilon < 1$ there exists an averaging sampler for the domain $\{0,1\}^m$ that uses $r = m + O\left (\log(1/\delta) \right )$ random bits and $t = O\left(\frac{1}{\varepsilon^2} \log \frac{1}{\delta} \right )^{1 + \beta}$ random samples. This sampler can be constructed in time $\poly(t, r) = \textnormal{poly}(m, \log(1/\delta), 1/\varepsilon)$.
\end{theorem}
After the sampling step, we will have to test if the sampled polynomials are irreducible. For this, we will use the following well known algorithm.
\begin{algorithm}\label{thm:test_irreducibility_algorithm}
    Let $\mathbb{F}_q$ be a finite field. There exists a deterministic algorithm that on input a polynomial 
    $f \in \mathbb{F}_q[w]$ of degree $a$, decides whether $f$ is irreducible in time $\poly(\log q, a)$.
\end{algorithm}
\begin{proof}[Proof (sketch)]
The algorithm is based on the classical characterization of irreducible polynomials
over finite fields. For each integer $i = 1, \ldots, a-1$, compute
\[
\gcd\!\bigl(f(x), x^{q^i} - x\bigr).
\]
Since every irreducible polynomial of degree $i$ over $\FF_q$ divides $x^{q^i} - x$,
the polynomial $f$ is irreducible over $\mathbb{F}_q$ if and only if all these greatest
common divisors are equal to $1$.

Each computation of $x^{q^i} \bmod f$ and the corresponding $\gcd$ can be carried out
in time polynomial in $\log q$ and $a$ (\cite{GathenGerhard2013}), and since the number of iterations is $a-1$,
the overall running time is $\poly(\log q, a)$.
\end{proof}
We are ready to construct $\mathbb{E}$ efficiently.
\begin{theorem}\label{thm:algorithm_constructing_E}
    Let $\FF = \FF_q$, let $k = 2^\ell$ be some positive power of $2$ and let $\delta > 0$. There exists a probabilistic algorithm that uses $O(k \log q + \log (1/\delta))$ random bits and runs in time $\poly(k \log q , \log (1/\delta))$, such that with probability at least $1-\delta$ outputs $\ell$ polynomials $p_i \in \FF[w_1,\ldots,w_\ell]$, such that $P = (p_1,\ldots, p_\ell) \lhd \FF[w_1,\ldots,w_\ell]$ is maximal and $\mathbb{E} = \FF[w_1,\ldots,w_\ell] / P$ is a field extension of degree $k$ over $\FF$. Moreover, for every \(1\le i\le\ell\), there exists
    $h_i\in\mathbb F_q^{\le i-1}[w_1,\ldots,w_{i-1}]$ such that
    $$
    p_i=
    \begin{cases}
    w_i^2-h_i, & \operatorname{char}\mathbb F_q\ne 2,\\
    w_i^2+w_i+h_i, & \operatorname{char}\mathbb F_q=2.
    \end{cases}
    $$
    Otherwise, the algorithm declares failure.
\end{theorem}
\begin{proof}
    First assume that $\text{char}\,\FF \neq 2$. We begin by setting the following notation.
    \begin{itemize}
        \item $R_i = \{f \in \FF_q[w_1,\ldots,w_i]: \deg_{w_j} f \leq 1 \textnormal{ for all } 1 \leq j \leq i \}$. Notice that $|R_i| \le q^{2^i}$.
        \item \(D=\prod_{i=1}^{\ell}(R_{i-1}\setminus\{0\})\), where
\(R_0=\mathbb F_q\).
        \item $m = \log |D| = O(k \log q)$.
        \item $\varepsilon = \frac{1}{2k}$.
    \end{itemize}
    For $\mathbf{a} = (\alpha_1,\ldots, \alpha_\ell) \in D$, let $\mathbf{p}(\mathbf{a}) = (p_1,\ldots,p_\ell) = (w_1^2 - \alpha_1,\ldots, w_\ell^2 - \alpha_\ell)$.
    Let $\FF_0 = \FF$. For $0 \leq i < \ell$, let $$\FF_{i+1} = \FF_i[w_{i+1}] / p_{i+1}(w_{i+1}).$$
    Let $\mathbb{E} = \FF_\ell$, and let $P = (p_1, \ldots, p_\ell) \lhd \FF[w_1, \ldots , w_\ell]$. Notice that $\FF_\ell$ is a field if and only if $P \lhd \FF_q[w_1,\ldots,w_\ell]$ is maximal, which happens if and only if $p_i \in \FF_{i-1}[w_i]$ is irreducible for all $1 \leq i \leq \ell$.

We will use a sampler to find $\mathbf{a} \in D$ such that $\mathbb{E}$ is a field.
Let $S \colon \{0,1\}^r \rightarrow D^t$ be the $(\delta,\eps)$ sampler given by \cref{thm:sampler construction} with $\beta = 1$, $r = m + O(\log \frac{1}{\delta})$ and $t = O\left(\frac{1}{\varepsilon^2} \log (1 / \delta) \right )^{2}$. Let $f \colon D \rightarrow \{0,1\}$ be the indicator function such that $f(\mathbf{a})=1$ if and only if $\mathbf{p}(\mathbf{a})$ is maximal. It is well known that for every finite field of characteristic different from $2$, the probability that an invertible element is a square equals $1/2$.
Hence, if $(\alpha_1,\ldots, \alpha_\ell) \in D$ is chosen uniformly, the probability that all polynomials $p_i \in \FF_{i-1}[w_i]$ are irreducible is $1/2^\ell = 1/k$. Hence, $\EE[f] = 1/k$.

Using the sampler $S$ we pick $t = O(k^2 \log(1/\delta))^{2}$ values $(\mathbf{a}_1,\ldots,\mathbf{a}_t) \in D^t$ using $r=m+O(\log(1/\delta))$ random bits in time $\poly(k \log (q), \log (1/\delta))$, such that
\[
    \Pr_{(Z_1,\ldots,Z_t) \sim S(U_r)}\left[ \left |\frac{1}{t} \sum_i f(Z_i) - \mathbb{E}[f] \right |  \leq \varepsilon \right] \geq 1 - \delta.
\]
Since $\mathbb{E}[f] = \frac{1}{k}$ and $\varepsilon = \frac{1}{2k}$, this implies in particular that with probability at least $1-\delta$ over $(Z_1,\ldots,Z_t) \sim S(U_r)$, $f(Z_i)=1$ for at least one $i \in [t]$.

On input a sampler seed $z \in \{0,1\}^r$, we compute $S(z) = (z_1,\ldots,z_t)$, and for each $z_i = (\alpha_1,\ldots,\alpha_\ell) \in D$, we test if $f(z_i)=1$, as follows:
\begin{itemize}
    \item For all $i$ from $1$ to $\ell$, do:
    \begin{itemize}
        \item Run \cref{thm:test_irreducibility_algorithm} to check if $p_i(w_i) \in \FF_{i-1}[w_i]$ is irreducible. If not, return false.
    \end{itemize}
    \item Return true.
\end{itemize}
If we have found that indeed $f(z_i) = 1$ for some $i \in [t]$, return $(p_1,\ldots, p_\ell)$. Otherwise, the algorithm declares failure. 

Recall that the runtime of each irreducibility test is at most $\poly\left(\log |\FF_{i-1}|\right) = \poly(k\,\log q)$. Thus, the overall runtime of our algorithm is $\poly(k,\log q,\log(1/\delta))$.
Finally, recall that if $f(\mathbf{a}) = 1$ then each $p_i \in \FF_{i-1}[w_i]$ is an irreducible polynomial of degree $2$, and hence $[\FF_{i} : \FF_{i-1}] = 2$. Therefore, $$[\mathbb{E}:\FF] = \prod_{1 \leq i \leq \ell} [\FF_{i}:\FF_{i-1}] = 2^\ell = k.$$ This completes the proof for the case of odd characteristic.

In characteristic \(2\), we instead take
\(D=\prod_{i=1}^{\ell}R_{i-1}\) and define
\(p_i=w_i^2+w_i+\alpha_i\).
Conditioned on the preceding extension steps succeeding,
\(p_i\) is irreducible over \(F_{i-1}\) if and only if the trace
of \(\alpha_i\) to \(\mathbb F_2\) equals \(1\).
For uniformly chosen \(\alpha_i\), this occurs with probability
\(1/2\), and the same analysis applies.
\end{proof}

\subsection{The PHSG Construction}
In this subsection we prove the following theorem.
\begin{theorem}\label{thm:HSG for polys construction}
    Let $\FF=\FF_q$ be a finite field, $d$ a positive integer and $\delta > 0$. Let $k = 2^\ell$ be a power of $2$. Then, there exists an absolute constant $c$ such that if $\delta \geq c \cdot \frac{d}{q^k}$, there exists a PHSG $H \colon T \rightarrow (\FF^{\leq \ell}[w_1,\ldots,w_\ell])^n$ with density $1 - \delta$ for polynomials of degree at most $d$ over $\FF$, which can be constructed in time $\poly(d, n, k,\log q)$, and has seed length $O(d \log n + k \log q)$.
\end{theorem}
\begin{proof}
    We start by picking $p_1,\ldots,p_\ell \in \FF[w_1,\ldots,w_\ell]$ using the algorithm from \cref{thm:algorithm_constructing_E} with $\delta' = \delta/2$. This algorithm uses $O(k \log q + \log (1/\delta)) = O(k\log q)$ random bits; i.e.\ there exists an efficiently computable map $$A \colon T_1 \rightarrow (\FF_q^{\leq \ell}[w_1,\ldots,w_\ell])^\ell$$ such that 
    with probability at least $1 - \delta'$ over $t \in T_1$ we have that $$A(t) = (p_1,\ldots,p_\ell) \in \FF[w_1,\ldots,w_\ell]$$
    satisfies that \(P=(p_1,\ldots,p_\ell)\) is a maximal ideal of
    \(\mathbb F[w_1,\ldots,w_\ell]\), and
    $$
    p_i=
    \begin{cases}
    w_i^2-h_i, & \operatorname{char}\mathbb F\ne 2,\\
    w_i^2+w_i+h_i, & \operatorname{char}\mathbb F=2,
    \end{cases}
    \qquad
    h_i\in\mathbb F^{\le i-1}[w_1,\ldots,w_{i-1}].
    $$
    Moreover, \(E=\mathbb F[w_1,\ldots,w_\ell]/P\) is a finite field
    with \(q^k\) elements, and \(\log|T_1|=O(k\log q)\).

    If the algorithm from \cref{thm:algorithm_constructing_E} did not declare failure,
    we can construct the field 
    \[\mathbb{E} = \FF[w_1,\ldots,w_\ell] / (p_1,\ldots,p_\ell).\]
    Let $\widehat{H} \colon T_2 \rightarrow\mathbb{E}^n$ be the \cite{GXHittingSets} HSG given by \cref{thm:HSG construction} for $n$-variate polynomials of degree at most $d$ over $\mathbb{E}$, set with $\delta' = \delta/2$. Recall that $\widehat{H}$ has seed length $O(d \log n + \log (1 / \delta)) = O(d \log n + k \log q)$.
    Let $$\varphi \colon \mathbb{E} \rightarrow \FF^{\leq \ell}[w_1,\ldots,w_\ell]$$ be the $\FF$-linear map such that $\varphi(\alpha)$ is the unique multilinear representative of the residue class.
    By \cref{claim:HSG for univariate polynomials}, $$H' = \varphi \circ \widehat{H} \colon T_2 \rightarrow (\FF_q^{\leq \ell}[w_1,\ldots,w_\ell])^n$$
    is a PHSG for $n$-variate polynomials of degree at most $d$ with density $1 - \delta'$. 
    
    Let \[H \colon T_1 \times T_2 \rightarrow (\FF_q[w_1,\ldots,w_\ell])^n\] be the map such that for $(t_1,t_2)\in T_1 \times T_2$, if $A(t_1)$ succeeds then $H(t_1,\cdot) = H'$, and otherwise, for concreteness, we set $H(t_1, \cdot) = 0$. Now, fix some
    nonzero $f \in \FF[x_1,\ldots,x_n]$ of degree at most $d$.
    With probability at least $1 - \delta'$ over $t_1 \in T_1$, $A(t_1)$ succeeds, and conditioned on that, with probability at least $1-\delta'$ over $t_2 \in T_2$, $$f(H(t_1,t_2)) = f(H'(t_2)) \neq 0.$$
    Therefore,
    \[
        \Pr_{(t_1,t_2) \in T_1 \times T_2}[f(H(t_1,t_2)) \neq 0] \geq 1 - 2 \delta' = 1 - \delta,
    \]
    which completes the proof.
\end{proof}

%% file: construction.tex
\section{Our PRG Construction}\label{sec:our PRG construction}
Let $n,d$ be positive integers, let $q = p^u$ be a prime power with $q \geq C((d \log d)^4 / \varepsilon^2)$, for some universal constant $C$ to be determined later on, and let $\varepsilon > 0$. We now present the construction of our $\eps$-error PRG
\[
    G \colon S \longrightarrow \FF_q^{n+1}
\]
for polynomials $f \in \FF_q[\mathbf{x}, y] = \FF_q[x_1\ldots,x_n,y]$ of degree at most $d$. Let $c \geq 1$ be an absolute constant, larger than the constants appearing in \cref{thm:HSG construction}, \cref{thm:HSG for polys construction} and \cref{thm:indecomposable implies equidistribution}.
We will need one PHSG and one HSG for the construction.
\begin{enumerate}
    \item Let $k = 2^\ell$ be a power of $2$ such that $\left \lceil \frac{d}{\log q} \right \rceil + 1 < k \leq 2 \left \lceil \frac{d}{\log q} \right \rceil + 2$. Let \[H_1\colon T_1\longrightarrow (\FF^{\leq \ell}_{q}[w_1,\ldots,w_\ell])^n\] 
    be the PHSG given by \cref{thm:HSG for polys construction} for $n$-variate polynomials over $\FF_q$ of degree at most $d$, with $\delta_1 = c \cdot d/q^k \leq  c \cdot \frac{d}{2^dq}$, and seed length $O(d \log n + \log q)$. 
    \item Let \[H_2 \colon T_2\longrightarrow \FF_{q}^n\] be the HSG given by \cref{thm:HSG construction} for $n$-variate polynomials over $\FF_q$ of degree at most $3d-1$, with $\delta_2 = c \cdot d / q$ and seed length $O(d \log n + \log q)$.
\end{enumerate}

We are now ready to define $G$.
Let $S = T_1 \times T_2 \times \FF_q^\ell \times \FF_q \times \FF_q$. Define $G \colon S \rightarrow \FF_q^{n+1}$ by
\begin{equation}\label{def:the PRG G}
    \begin{split}
        G(r, s, \mathbf{t}, u, v) = (&H_1(r)_1(\mathbf{t}) \cdot v + H_2(s)_1 \cdot u + H_2(s)_1^2, \ldots, \\
        &H_1(r)_n(\mathbf{t}) \cdot v + H_2(s)_n \cdot u + H_2(s)_n^2, \ u),
    \end{split}
\end{equation}
where $\mathbf{t} = (t_1,\ldots,t_\ell)$.

Note that the running time of $G$ on input $s \in S$ is $\poly(n,d,\log q)$.
Indeed, by \cref{thm:HSG construction,thm:HSG for polys construction}, the HSG and PHSG can be computed within this time bound, and the additional step of evaluating $\ell = O(\log d)$ variables in $n$ polynomials of degree at most $\ell$ requires only $\poly(n,d,\log q)$ time.

We proceed by showing that if $f$ is indecomposable, then the random restricted polynomial $F = f \circ \mathbf{p}$ remains indecomposable with high probability.
\begin{proposition}\label{prop:preserve-indecomposability}
    \label{lem:reduction}
Let $f \in \FF_q[\mathbf{x},y]$ be an indecomposable $(n+1)$-variate polynomial of degree at most $d$ over $\FF_q$. Let $(r,s)$ be a random element of $T_1 \times T_2$. Let $H_2(s)=(a_1,\ldots,a_n)=\mathbf{a}$, let $H_1(r)=(b_1(\mathbf{w}),\ldots,b_n(\mathbf{w}))=\mathbf{b}(\mathbf{w})$ for $\mathbf{w}=(w_1,\ldots,w_\ell)$, and finally, denote
\[
    F=f(b_1(\mathbf{w})x+a_1y+a_1^2,\ldots,b_n(\mathbf{w})x+a_ny+a_n^2,y) \in \FF_q[x,y,\mathbf{w}].
\]
Then,
\[
    \Pr[F~\text{is indecomposable over}~\FF_q] \geq 1-\delta_2-(2^{d-1}-1)\delta_1.
\]
\end{proposition}
\begin{proof}
Recall that $s_{\mathbf{a}}$ fixes $y$ and sends $x_i$ to $x_i+a_i y+a_i^2$. As $f$ is indecomposable over $\FF_q$, so is $s_{\mathbf{a}}(f)$. By Lemmas~\ref{lemma:indecomposable over algebraic closure} and~\ref{lemma:indecomposable equivalence to irreducibility}, $s_{\mathbf{a}}(f)-t$ is irreducible over $\overline{\FF_q(t)}$.
$f$ is indecomposable and in particular not a $p$-th power, thus some partial derivative of $f$ is nonzero. By Lemma~\ref{cor:HSG for Hyptothesis (H)}, there exists a nonzero $B \in \FF_q[\mathbf{x}]$ of degree at most $3d-1$ such that $B(\mathbf{a}) \neq 0$ implies
\begin{equation}\label{eq:cdotg}
    s_{\mathbf{a}}(f)-t=c\cdot g,
\end{equation}
where $c \in \FF_q^\times$ and
\[
    g \in \FF_q(t)[\mathbf{x}, y] \subseteq \overline{\FF_q(t)}[\mathbf{x}, y]
\]
is a polynomial of degree at most \(d\) satisfying Hypothesis~(H). Since $H_2$ is a HSG, the event $B(a) \neq 0$ happens with probability at least $1-\delta_2$. Condition on this event, so \cref{eq:cdotg} holds.
As $s_\mathbf{a}(f)-t$ is irreducible over $\overline{\FF_q(t)}$, so is $g$.

Let $\KK = \overline{\FF_q(t)}$  and let $\mathbb{L} = \overline{\FF_q(t, \mathbf{w})}$ be such that $w_1,\ldots,w_{\ell}$ are new variables.
Let $m=2^{d-1}-1$.
As $g \in \KK[\mathbf{x},y]$, by \cref{prop:HSG property for irreducibility}, there exist nonzero polynomials $Q_1,\dots,Q_m\in \KK[z_1,\dots,z_n]$ of degree at most $d$ such that the union of the zero loci of these polynomials contains all $\mathbf{b}^*=(b_1^*,\dots,b_n^*)\in \mathbb{L}^n$ for which
$g(b_1^* x, \dots, b_n^* x, y)$ is reducible over $\mathbb{L}$.
$H_1$ is a PHSG with density $1-\delta_1$ for polynomials of degree at most $d$ over $\overline{\FF_q(t)}$.
Therefore, for each $i\in [m]$, the probability that $Q_i(\mathbf{b}(\mathbf{w}))=0$ is at most $\delta_1$.

Condition on the event that $Q_i(\mathbf{b}(\mathbf{w})) \neq 0$ for all $i \in [m]$. Then, $g(b_1(\mathbf{w}) x, \dots, b_n(\mathbf{w}) x,y)$ is \emph{irreducible} over $\mathbb{L}$.
On the other hand, note that
\begin{align*}
    c\cdot g(b_1(\mathbf{w}) x, \dots, b_n(\mathbf{w}) x,y)
    &\stackrel{\eqref{eq:cdotg}}{=}(s_{\mathbf{a}}(f))(b_1(\mathbf{w}) x,\dots,b_n(\mathbf{w}) x,y)-t \\
    &\ =f(b_1(\mathbf{w}) x+a_1 y + a_1^2,\dots,b_n(\mathbf{w}) x+a_n y+a_n^2,y)-t=F-t,
\end{align*}
where the second step uses the definition
$s_{\mathbf{a}}(f)=f(x_1+a_1 y + a_1^2,\dots,x_n+a_n y+a_n^2, y)\in\FF_q[\mathbf{x},y]$. 

Thus, $F-t$ is irreducible over $\mathbb{L}$, and hence as an element in $\overline{\FF_q(t)}(\mathbf{w})[x, y]$. By \cref{lemma:Gauss}, it is then irreducible as an element in $\overline{\FF_q(t)}(\mathbf{w}, x)[y]$.
The leading coefficient of \(F-t\) as a polynomial in \(y\)
is \(c\in\mathbb F_q^\times\).
Thus, \(F-t\), viewed as a polynomial in \(y\) over
\(\overline{\mathbb F_q(t)}[\mathbf w,x]\), is primitive. Hence, by \cref{lemma:multivariate Guass}, $F-t$ is irreducible as an element in $$\overline{\FF_q(t)}[\mathbf{w},x][y] = \overline{\FF_q(t)}[x, y, \mathbf{w}].$$
By \cref{lemma:indecomposable equivalence to irreducibility}, $F$ is indecomposable over $\overline{\FF}_q$. So it is indecomposable over $\FF_q$.

Overall, note that 
the indecomposability of $F$ over $\FF_q$ relies on the conditions $B(\mathbf{a})\neq 0$ and $Q_1(\mathbf{b}(\mathbf{w})) \neq 0,\dots,Q_m(\mathbf{b}(\mathbf{w}))\neq 0$. By the union bound, these conditions are simultaneously satisfied with probability at least $1-\delta_2-m\delta_1=1-\delta_2 - (2^{d-1} - 1)\delta_1$, which completes the proof.
\end{proof}

\begin{theorem}\label{thm:main}
There exists an absolute constant $C>0$ such that for all $\eps>0$ and $q\geq C\frac{(d\log d)^4}{\eps^2}$,
$G$ as defined in \cref{def:the PRG G} is a PRG for $(n+1)$-variate polynomials of degree at most $d$ over $\mathbb{F}_q$ with error $\eps$ and seed length $O(d\log n+\log q)$.
\end{theorem}

\begin{proof} 
Let $C = 16^3\cdot c^2$ for the absolute constant $c$ defined earlier.
Let $f\in\FF_q[\mathbf{x}, y]$ be a polynomial of degree at most $d$.
We want to prove that $f(G(\mathbf{U}_S))$ and $f(\mathbf{U}_{\FF_q^{n+1}})$ are $\varepsilon$-close in statistical distance.
We may assume that $f$ is non constant, since the claim is trivial otherwise.

Our first step is the same as in \cite{DV22} and \cite{dwivedi2024optimalpseudorandomgeneratorslowdegree}: $f$ can always be written in the form $f=g \circ h$, where $g\in\FF_q[z]$ is a univariate polynomial and $h\in\FF_q[\mathbf{x},y]$ is indecomposable over $\FF_q$. Let $D=h(G(\mathbf{U}_S))$ and $D'=h(\mathbf{U}_{\FF_q^{n+1}})$. Then $f(G(\mathbf{U}_S))=g(D)$ and $f(\mathbf{U}_{\FF_q^{n+1}})=g(D')$. If $D$ and $D'$ are $\eps$-close, then $g(D)$ and $g(D')$ are also $\eps$-close. Thus, by replacing $f$ with $h$, we may assume that $f$ is indecomposable over $\FF_q$.

Let $r,s,\mathbf{a},\mathbf{b}$ and $F$ be as in \cref{lem:reduction}. Then, by \cref{lem:reduction}, the probability that $F$ is decomposable over $\FF_q$ over a random choice of $r$ and $s$ is at most $2^{d-1}\delta_1 + \delta_2 \leq 2c \cdot d/q$.
Fix $r$ and $s$ such that $F$ is indecomposable over $\FF_q$ and note that by definition,
\[
    f(G(r,s,\mathbf{t},u,v))=F(v,u,\mathbf{t}).
\] 
Applying \cref{thm:indecomposable implies equidistribution} to $F$ shows that, for such fixed $r$ and $s$, the distribution of $F(v,u,\mathbf{t})$, i.e., $f(G(r,s,\mathbf{t},u,v))$, over random $t_i,u,v\in\FF_q$ is $\eps'$-close to $\mathbf{U}_{\FF_q}$, where $\eps'\leq c \cdot  (\deg F)^2/\sqrt{q}$. Notice that
\[
    \deg F  \leq d \cdot (\ell + 1) \leq d \cdot 4\log d.
\]
This implies that the statistical distance between $f(G(\mathbf{U}_S))$ and $\mathbf{U}_{\FF_q}$ is at most
$$
2c \cdot \frac{d}{q} + 16c \cdot \frac{(d \log d)^2}{\sqrt{q}} \leq \varepsilon/2.
$$

On the other hand, as $f$ is also indecomposable over $\FF_q$, applying \cref{thm:indecomposable implies equidistribution} to $f$ shows that $f(\mathbf{U}_{\FF_q^{n+1}})$ is $\varepsilon'$-close to $\mathbf{U}_{\FF_q}$, where $\varepsilon' = c \cdot d^2 / \sqrt{q} \leq \varepsilon/2$.
Therefore, the statistical distance between $f(G(\mathbf{U}_S))$ and $f(\mathbf{U}_{\FF_q^{n+1}})$ is at most $\varepsilon$.

The seed length of $G$ is
\[
\log|T_1| + \log|T_2| + \ell \log q + 2 \log q =O(d\log n+\log q),
\]
which completes the proof. 
\end{proof}

%% file: reduction.tex
\section{An Approach Towards Smaller Fields}\label{sec:f2}

In this section, we prove \cref{thm:f2}. We begin by introducing some notation. For $q=p^{a}$, where $p$ is prime, denote by $\Tr_{q \rightarrow p} \colon \FF_q \rightarrow \FF_p$ the absolute field trace. Assume that we are given a PRG $G_q \colon S \rightarrow \FF_q^n$ for $n$-variate polynomials of total degree at most $D$ over $\FF_q$, with error $\eps$. For simplicity, we treat $\eps$ as a constant throughout. We note, however, that the seed length obtained through our reduction depends polynomially on $1/\eps$, rather than logarithmically as desired.

Our proposed reduction for constructing a PRG for polynomials over smaller fields is to take traces. Namely, we define
\(
G_p \colon S \rightarrow \FF_p^n
\)
by
\begin{equation}\label{eq:gq to gp}
G_p(s) = (\Tr_{q \rightarrow p} \circ G_q)(s)
= \left( \Tr_{q \rightarrow p}(G_q(s)_1),\ldots,\Tr_{q \rightarrow p}(G_q(s)_n) \right).
\end{equation}
The following lemma states the reduction. We then choose the parameters to obtain \cref{thm:f2}.

\begin{lemma}\label{prop:main-reduction}
Let $p$ be prime and let $q$ be a power of $p$. Suppose that $G_q\colon S\rightarrow\FF_q^n$ is a PRG for $n$-variate polynomials of total degree at most $dq/p$ over $\FF_q$, with error $\eps$. Then
\[
G_p=\Tr_{q\rightarrow p}\circ G_q
\]
is a PRG for $n$-variate polynomials of total degree at most $d$ over $\FF_p$, with the same error and seed length.
\end{lemma}

\begin{proof}
For any $f\in\FF_p^{\le d}[x_1,\ldots,x_n]$, the composition
\[
h=f\circ\Tr_{q\rightarrow p},
\]
where the trace is applied coordinatewise as in \cref{eq:gq to gp}, has degree at most $dq/p$ as a polynomial over $\FF_q$.  Since
\[
f\circ G_p
= f\circ(\Tr_{q\rightarrow p}\circ G_q)
= (f\circ\Tr_{q\rightarrow p})\circ G_q
= h\circ G_q,
\]
the guarantee for $G_q$ implies that $G_p$ fools $f$ with error $\eps$. Indeed, for every $A\subseteq\FF_p$,
\[
\begin{aligned}
&\left|
\mathbb{E}_{s\sim S}\!\left[\mathbf{1}_{A}(f(G_p(s)))\right]
-
\mathbb{E}_{x\sim\FF_p^n}\!\left[\mathbf{1}_{A}(f(x))\right]
\right|\\
&\qquad=
\left|
\mathbb{E}_{s\sim S}\!\left[\mathbf{1}_{A}(h(G_q(s)))\right]
-
\mathbb{E}_{y\sim\FF_q^n}\!\left[\mathbf{1}_{A}(h(y))\right]
\right|
\le\eps,
\end{aligned}
\]
where the equality uses the composition identity above and the fact that the coordinatewise trace maps the uniform distribution on $\FF_q^n$ to the uniform distribution on $\FF_p^n$.
\end{proof}

With \cref{prop:main-reduction} in hand, we are ready to  derive \cref{thm:f2}.

\begin{proof}[Proof of \cref{thm:f2}]
Fix a constant $\delta\in(0,1)$ and a prime $p\le d$. A PRG $G_q$ with field-size requirement $q\ge D^\delta$, where $D$ is its degree parameter, can be used in the reduction provided that
\[
q\ge\left(\frac{dq}{p}\right)^\delta,
\]
or, equivalently,
\[
q\ge\left(\frac{d}{p}\right)^{\frac{\delta}{1-\delta}}.
\]
Choose $q$ to be the smallest power $p^a$, with $a\ge1$, satisfying this inequality. Then
\[
q\le p\left(\frac{d}{p}\right)^{\frac{\delta}{1-\delta}}.
\]
Consequently, an optimal-seed PRG $G_q$ for degree $D=dq/p$ gives a PRG over $\FF_p$ with seed length
\[
O(D\log n+\log q)
=
O\!\left(d\left(\frac{d}{p}\right)^{\frac{\delta}{1-\delta}}\log n\right).
\]
In particular, setting $p=2$ gives seed length $O(d^{1/(1-\delta)}\log n)$, as claimed.
\end{proof}

%% file: ai.tex
\section*{AI Usage}
This disclosure is included as part of the TCS community's initial efforts to adapt to AI use, and in particular in response to the recommendations of the September 2026 report \emph{AI and TCS: The Next Six Months}~\cite{AITCS26}.

All ideas and proofs were developed by the authors without AI assistance, except for the argument that allows us to handle arbitrary characteristic. We thank Kai Zhe Zheng for informing us that, using our earlier manuscript as input, he had previously obtained an AI-generated argument for a similar characteristic-free result. Our argument was developed independently, without seeing that proof. We also thank Dean Doron for his collaboration and substantial contributions to this work.